\documentclass[11pt, a4paper]{article}

\usepackage[T1]{fontenc}
\usepackage[utf8]{inputenc}
\usepackage[american]{babel}
\usepackage{url}
\usepackage{hyperref}
\usepackage{authblk}
\usepackage{setspace}
\usepackage{abstract}
\usepackage{xcolor}
\usepackage{fancyhdr}
\usepackage[top=1in, left=1in, bottom=1in, right=1in]{geometry}
\usepackage[sectionbib,round]{natbib}
\usepackage{enumerate}
\usepackage{setspace}
\usepackage{mathptmx}
\usepackage{multicol}
\usepackage{multirow}
\usepackage{caption}
\usepackage{subcaption}
\usepackage{amsfonts}
\usepackage{amsmath}
\usepackage{upgreek}
\usepackage{amsthm}
\usepackage{amssymb}
\usepackage{mathtools}
\usepackage{threeparttable}
\usepackage{graphicx}
\usepackage{float}
\usepackage{booktabs}
\usepackage{pgfplots}
\pgfplotsset{compat=1.12}
\usepackage{enumitem}
\usepackage{dirtytalk}
\usepackage{tikz}
\usetikzlibrary{positioning}

\newtheorem{proposition}{Proposition}
\newtheorem{definition}{Definition}
\newtheorem{corollary}{Corollary}
\newtheorem{lemma}{Lemma}
\newtheorem{remark}{Remark}
\newtheorem{assumption}{Assumption}

\hypersetup{
    colorlinks=true,
    allcolors=blue,
    linkcolor=blue,
    urlcolor=blue,
    citecolor=blue,
    pdfborder={0 0 0}
}

\fancypagestyle{maincontent}{
    \fancyhf{}
    
    \fancyfoot[C]{\thepage}
}

\begin{document}

\begin{titlepage}
    \title{Multiplicity of Equilibria in the War of Attrition \\ with Two-Sided Asymmetric Information}
    \author[1]{Martin Castillo-Quintana}
    \author[2]{Gianfranco Miranda-Romero}
    \affil[1]{Harris School of Public Policy, University of Chicago}
    \affil[2]{National Bureau of Economic Research}
    \date{\today}
    \maketitle
    \thispagestyle{empty}

    \begin{abstract}
        The war of attrition with two-sided asymmetric information is a foundational
        model in political economy, yet it generically admits a continuum of perfect
        Bayesian equilibria.  This paper characterizes the sources of equilibrium
        multiplicity. We identify conditions
        on the type distribution that determine which form of multiplicity arises:
        when the lower limit of the \textit{hazard potential}---the integral of the
        hazard rate normalized by type---diverges, the free parameter is the relative
        aggressiveness of strategies; when that limit is finite, the free parameter
        is the mass of types conceding immediately.  We prove that the
        Amann--Leininger payoff perturbation and the introduction of behavioral
        types---two seemingly distinct refinements---are mathematically equivalent
        and succeed in selecting a unique equilibrium if and only if the type support
        is bounded.  For unbounded supports, multiplicity persists. These results provide guidance for applied theorists: choosing
        distributions with bounded support ensures existing refinements deliver
        unique predictions.
    \end{abstract}


    \vspace*{\fill}
\end{titlepage}

\clearpage
\pagestyle{maincontent}

\section{The Strategic Problem}\label{sec:intro}

\begin{quote}
\emph{\say{What we are seeing play out here is essentially a war of attrition.}}

\hfill ---Brigadier General Steve Anderson, retired, on the \\
\null\hfill US--Israeli military escalation against Iran \citeyearpar{cnn2026generals}
\end{quote}

The war of attrition with two-sided incomplete information is a workhorse model in
political economy and game theory.  Applications span the emergence of property
rights \citep{hafer2006origins}, exit from duopolies
\citep{fudenberg1986theory,takahashi2015estimating}, interstate and civil conflicts
\citep{reich2023dynamic,fearon1994domestic,myerson2023game}, and reputational
bargaining
\citep{abreu2000bargaining,kambe1999bargaining,wolitzky2012reputational,fanning2016reputational}.

The strategic tension is elegant: each side screens the other by refusing to
concede, learning over time that remaining opponent types are stronger.  Each type
quits when the marginal cost of continued screening exceeds its marginal benefit.
This two-sided screening structure, however, generically produces multiple perfect
Bayesian equilibria \citep{nalebuff1985asymmetric,fudenberg1991game}.

The dominant fix in the literature introduces a positive mass of \say{crazy} or
\say{behavioral} types who fight forever
\citep{myerson2013game,powell2017taking,fanning2022reputational}.  This approach
simplifies equilibrium analysis and often delivers uniqueness.  Yet it abandons a
core principle of game theory: the study of rational actors in strategic contexts.

In a complementary approach, \citet{myatt2025impact} obtains uniqueness
by imposing a finite deadline on stopping times and characterizes how
asymmetries in perceived strength---ranked by hazard-rate dominance across
players' type distributions---determine which side concedes.  As the
deadline recedes, the perceived weaker player concedes immediately.  His
results extend to the hybrid all-pay auction and crazy-types specifications,
delivering uniqueness for general distributions under either modification.
Our analysis differs in focus: rather than characterizing the selected
equilibrium's properties, we characterize the \textit{sources} of
multiplicity in the unperturbed game through the hazard potential and
identify the precise conditions under which the leading perturbation-based
refinements succeed or fail.

This paper pursues a different path.  We characterize the sources of equilibrium
multiplicity without behavioral types and identify conditions under which existing
refinements---the Amann--Leininger \citeyearpar{amann1996asymmetric} payoff
perturbation and the introduction of behavioral types---succeed or fail at selecting
a unique equilibrium.

We establish three results.  First, we provide a complete characterization of
equilibrium multiplicity in terms of a single function $\Lambda$---the
\textit{hazard potential} of the type distribution (defined in
Section~\ref{sec:char})---constructed from the hazard rate of the type distribution.
The behavior of $\Lambda$ near the boundaries of the type support determines which
form of multiplicity arises: a continuum indexed by the relative aggressiveness of
strategies (when the lower limit of $\Lambda$ diverges) or a continuum indexed by
the mass conceding immediately (when the lower limit is finite).

Second, we prove that the Amann--Leininger perturbation and the behavioral-types
refinement are mathematically equivalent: they modify the equilibrium differential
equation in the same way and succeed or fail under identical conditions.

Third, we establish that both refinements select a unique equilibrium if and only if
the type support is bounded.  For unbounded supports---including the exponential and
Pareto distributions commonly used in applications---multiplicity persists under any
perturbation.

Section~\ref{sec:model} presents the model and results on equilibrium
properties.  Section~\ref{sec:examples} develops three canonical examples
illustrating distinct forms of multiplicity.  Section~\ref{sec:char} provides the
general characterization through the hazard potential.  Section~\ref{sec:refinements}
introduces the refinements and proves the main selection theorem.
Section~\ref{sec:discussion} discusses implications for applied theory and directions
for future work.

\section{The Model}\label{sec:model}

\subsection{Setup}

Two players $i \in \{1,2\}$ fight over an indivisible object.  Player $i$ values
the object at $\theta_i$, her privately known type.  Both players believe the
opponent's type is drawn independently from an absolutely continuous distribution
$F$ with support $\Theta = (\underline{\theta}, \overline{\theta})$, where
$0 \leq \underline{\theta} < \overline{\theta} \leq \infty$.

\begin{assumption}[Density Regularity]\label{ass:density}
The density $f$ is continuous and strictly positive on
$(\underline{\theta}, \overline{\theta})$.
\end{assumption}

This assumption is satisfied by all standard parametric families---including the
exponential, uniform, and Pareto distributions analyzed in
Section~\ref{sec:examples}---and is maintained throughout.

Each player simultaneously chooses a stopping time from
$\overline{\mathbb{R}}_+ = [0,\infty) \cup \{\infty\}$.  A \textit{strategy} for
player $i$ is a function $\sigma_i: \Theta \to \overline{\mathbb{R}}_+$.  For
stopping times $(a_1, a_2)$ and type $\theta_i$, payoffs are
\[
u_i(a_i, a_{-i}, \theta_i) = \begin{cases}
\theta_i - a_{-i} & \text{if } a_i > a_{-i} \\[3pt]
\dfrac{\theta_i}{2} - a_i & \text{if } a_i = a_{-i} \\[3pt]
-a_i & \text{if } a_i < a_{-i}.
\end{cases}
\]
The winner receives her prize valuation minus the loser's stopping time; the loser
pays her own stopping time; ties split the prize with each player paying her own
stopping time.  This specification corresponds to an ascending second-price all-pay
auction \citep{klemperer1999auction}: the price rises until one player concedes, and
both pay the exit price.  The cost structure---where the winner pays the loser's
stopping time---captures the feature that fighting costs accrue only until
concession.

The solution concept is Bayesian Nash equilibrium.  All results extend to perfect
Bayesian equilibrium in the continuous-time dynamic formulation.

\begin{remark}[Pure Strategies]\label{rem:pure}
This game admits no Bayesian Nash equilibrium in non-degenerate mixed strategies.
The argument proceeds in two steps that avoid circular dependence on the
pure-strategy characterization.

\textit{Step~1: Atomless stopping-time distributions.}  In any BNE---whether in
pure or mixed strategies---the distribution of each player's stopping time is
atomless on $(0, \infty)$.  Suppose Player~$i$'s stopping time places an atom of
mass $p > 0$ at some $a^* \in (0, \infty)$.  Any Player~$j$ type with $\theta_j > 0$
who would otherwise stop at $a^*$ can deviate to $a^* + \varepsilon$, beating the
atom and gaining approximately $\theta_j p - \varepsilon > 0$ for small
$\varepsilon$.  Since $F$ is atomless and $\theta_j > 0$ with positive probability,
such types exist, contradicting the optimality of the atom.  An analogous argument
(that of Lemma~\ref{lem:zero} below) shows that at most one player's stopping-time
distribution can place an atom at $a = 0$.  Neither argument presupposes
monotonicity, continuity, or purity of the underlying strategies.

\textit{Step~2: Unique pure best response.}  Given that the opponent's stopping-time
distribution $G_{-i}$ is atomless on $(0, \infty)$, the expected payoff of type
$\theta_i$ choosing stopping time $a_i > 0$ is
\[
\pi_i(a_i, \theta_i) = \theta_i\, G_{-i}(a_i)
    - \int_0^{a_i} \bigl(1 - G_{-i}(s)\bigr)\, ds.
\]
The derivative with respect to $a_i$,
$\pi_i'(a_i, \theta_i) = \theta_i\, g_{-i}(a_i) - \bigl(1 - G_{-i}(a_i)\bigr)$,
sets $\theta_i$ equal to the reciprocal of the hazard rate $h_{-i}(a_i) :=
g_{-i}(a_i)/\bigl(1 - G_{-i}(a_i)\bigr)$ of the opponent's stopping time.
Since $\pi_i(0, \theta_i) \geq 0$ and $\pi_i(a_i, \theta_i) \to -\infty$ as
$a_i \to \infty$ (the cost $\int (1-G_{-i})$ grows without bound while
$\theta_i G_{-i}(a_i) \leq \theta_i$), a global maximum exists.  Moreover,
the payoff satisfies the single-crossing property in $(\theta_i, a_i)$: for
$a' > a$ and $\theta' > \theta$,
\[
\bigl[\pi_i(a', \theta') - \pi_i(a, \theta')\bigr]
  - \bigl[\pi_i(a', \theta) - \pi_i(a, \theta)\bigr]
= (\theta' - \theta)\bigl[G_{-i}(a') - G_{-i}(a)\bigr] > 0,
\]
so the optimal stopping time is non-decreasing in $\theta_i$.  Combined with the
existence of a global maximum for each $\theta_i$, the first-order condition
pins down a unique optimal stopping time for each type.  Any mixed strategy is
therefore dominated by the corresponding pure best response.
\end{remark}

\subsection{Preliminary Results}

We collect standard results on equilibrium properties.  All proofs appear in
Appendix~\ref{app:proofs}.

\begin{lemma}[Monotonicity and Continuity; \citealt{fudenberg1986theory}]%
\label{lem:monotone}
Let $(\sigma_1, \sigma_2)$ be an equilibrium.  Wherever $\sigma_1$ and $\sigma_2$
are finite, they are (i)~non-decreasing and (ii)~continuous.
\end{lemma}

Monotonicity follows because higher-valuation types value winning more: if type
$\theta$ optimally bids $a^*$, type $\theta + \varepsilon$ cannot optimally bid
below $a^*$.  Continuity follows because a jump in $\sigma_1$ creates an interval of
stopping points that no Player~1 type selects; Player~2 types just above this
interval pay the full gap for a negligible gain in winning probability, contradicting
optimality.

\begin{lemma}[Strict Monotonicity in the Interior]\label{lem:strict}
Let $(\sigma_1, \sigma_2)$ be an equilibrium.  If $\theta < \theta'$ with
$0 < \sigma_i(\theta)$ and $\sigma_i(\theta') < \infty$ for $i = 1,2$, then
$\sigma_i(\theta) < \sigma_i(\theta')$.
\end{lemma}

If $\sigma_1$ were flat at level $K$ over an interval, a positive mass of
Player~1 types would tie at $K$.  Player~2 could bid $K + \varepsilon$, beating this
entire mass for negligible extra cost.  But then $K$ is not in the range of
$\sigma_2$, contradicting continuity.

\begin{lemma}[At Most One Side Concedes at Zero; \citealt{hendricks1988war}]%
\label{lem:zero}
At most one player concedes immediately with positive probability:
\[
\min\!\left\{\sup_{\sigma_1(\theta)=0} \theta,\;
             \sup_{\sigma_2(\theta)=0} \theta\right\} = \underline{\theta}.
\]
\end{lemma}

If Player~1 has mass $p > 0$ at zero and Player~2 type $\theta_2$ also concedes at
zero, Player~2 can deviate to $\varepsilon$, beating all zero-types of Player~1 and
gaining approximately $\theta_2 p - \varepsilon > 0$ for small $\varepsilon$, a
contradiction.

\begin{lemma}[At Most One Side Fights Forever]\label{lem:infinity}
At most one player has a positive mass of types choosing $\infty$, in the sense that
\[
\max\{\overline{\theta}_1, \overline{\theta}_2\} = \overline{\theta},
\]
where $\overline{\theta}_i := \inf\{\theta : \sigma_i(\theta) = \infty\}$ is the
lowest type of player $i$ that fights forever.
\end{lemma}

If both players had positive-mass sets of types choosing $\infty$, any such type
would tie at $\infty$ with positive probability.  Since the tie payoff is
$\theta_i/2 - \infty = -\infty$, this is dominated by any finite stopping point.

\begin{lemma}[Differentiability]\label{lem:diff}
Let $(\sigma_1, \sigma_2)$ be an equilibrium.  Then $\sigma_1$ and $\sigma_2$ are
differentiable wherever they take values in $(0, \infty)$.
\end{lemma}

Since $\sigma_i$ is monotone and continuous on the interior
(Lemmas~\ref{lem:monotone}--\ref{lem:strict}), it is differentiable almost
everywhere by Lebesgue's theorem.  On each interval of differentiability, the
first-order conditions define an ODE whose right-hand side is locally Lipschitz
continuous under Assumption~\ref{ass:density}.  By the Picard--Lindel\"of theorem,
the solution to this ODE is unique and $C^1$ on each such interval.  Local Lipschitz
continuity of $\sigma_i$ (established in the proof; see Appendix~\ref{app:proofs})
then implies that the right-hand side of the ODE is locally Lipschitz in $\theta$
as well, so the unique $C^1$ solution extends continuously across any isolated
non-differentiable point, yielding differentiability on the entire interior.

\subsection{The Type-to-Type Mapping}\label{sec:k}

The preceding lemmas establish that in any equilibrium, strategies are strictly
increasing and differentiable in the interior.

\begin{assumption}\label{ass:wlog}
Without loss of generality, only Player~1 may concede at zero with positive
probability:
\[
\underline{\theta}_1 := \sup\{\theta : \sigma_1(\theta) = 0\} \geq \underline{\theta},
\qquad
\sup\{\theta : \sigma_2(\theta) = 0\} = \underline{\theta}.
\]
This is without loss by Lemma~\ref{lem:zero} and relabeling of players.
\end{assumption}

Recall that $\overline{\theta}_i := \inf\{\theta : \sigma_i(\theta) = \infty\}$
denotes the lowest type of player $i$ that fights forever (Lemma~\ref{lem:infinity}).
Following \citet{amann1996asymmetric}, define the \textit{type-to-type mapping}
\[
k(\theta, \underline{\theta}_1) := \sigma_2^{-1}(\sigma_1(\theta, \underline{\theta}_1)),
\quad \theta \in (\underline{\theta}_1, \overline{\theta}_1),
\]
which maps each type $\theta \in (\underline{\theta}_1, \overline{\theta}_1)$ of
Player~1 to the type $k(\theta, \underline{\theta}_1) \in (\underline{\theta},
\overline{\theta}_2)$ of Player~2 that chooses the same stopping point.  By
construction, $k(\underline{\theta}_1, \underline{\theta}_1) = \underline{\theta}$.\footnote{Throughout, $g(\underline{\theta})$ and $g(\overline{\theta})$ denote the right-hand
limit $\lim_{\theta \downarrow \underline{\theta}} g(\theta)$ and the left-hand limit
$\lim_{\theta \uparrow \overline{\theta}} g(\theta)$, respectively.}

\begin{lemma}[ODE Characterization; \citealt{amann1996asymmetric}]\label{lem:ode}
In any equilibrium, the
type-to-type mapping $k(\cdot, \underline{\theta}_1)$ and boundary value
$\underline{\theta}_1$ satisfy
\begin{align}
\frac{\partial k(\theta, \underline{\theta}_1)}{\partial \theta}
    &= \frac{k(\theta,\underline{\theta}_1)\bigl(1 - F(k(\theta,\underline{\theta}_1))\bigr)}%
           {f(k(\theta,\underline{\theta}_1))}
      \cdot \frac{f(\theta)}{\theta\bigl(1 - F(\theta)\bigr)}, \label{eq:kode}\\[6pt]
\frac{\partial \sigma_1(\theta, \underline{\theta}_1)}{\partial \theta}
    &= \frac{k(\theta, \underline{\theta}_1)\, f(\theta)}{1 - F(\theta)}, \label{eq:sigmaode}
\end{align}
with $\sigma_1(\theta, \underline{\theta}_1) = \sigma_2(k(\theta,\underline{\theta}_1))$ and
$k(\underline{\theta}_1, \underline{\theta}_1) = \underline{\theta}$.
\end{lemma}

\begin{definition}[Admissible Type-to-Type Function]\label{def:admissible}
A pair $(k, \underline{\theta}_1)$ with
$\underline{\theta}_1 \in [\underline{\theta}, \overline{\theta})$ and
$k: (\underline{\theta}_1, \overline{\theta}_1) \to
(\underline{\theta}, \overline{\theta}_2)$ is \textit{admissible} if:
\begin{enumerate}[label=(\roman*)]
  \item $k$ is strictly increasing and differentiable on
    $(\underline{\theta}_1, \overline{\theta}_1)$;
  \item $k$ satisfies the ODE \eqref{eq:kode} on
    $(\underline{\theta}_1, \overline{\theta}_1)$;
  \item $\lim_{\theta \downarrow \underline{\theta}_1}
    k(\theta, \underline{\theta}_1) = \underline{\theta}$;
  \item $\sigma_1(\theta, \underline{\theta}_1) :=
    \int_{\underline{\theta}_1}^{\theta}
    k(t, \underline{\theta}_1)\, f(t)/\bigl(1 - F(t)\bigr)\, dt < \infty$
    for all $\theta \in (\underline{\theta}_1, \overline{\theta}_1)$.
\end{enumerate}
\end{definition}

Once an admissible $(k, \underline{\theta}_1)$ is fixed,
$\sigma_1$ is uniquely determined by integration,
\[
\sigma_1(\theta, \underline{\theta}_1)
    = \int_{\underline{\theta}_1}^{\theta}
      \frac{k(t, \underline{\theta}_1)\, f(t)}{1 - F(t)}\, dt,
\]
and $\sigma_2$ follows from
$\sigma_2(\theta_2) = \sigma_1(k^{-1}(\theta_2, \underline{\theta}_1),
\underline{\theta}_1)$.

\begin{proposition}[Sufficiency]\label{prop:suff}
Let $(k, \underline{\theta}_1)$ be admissible in the sense of
Definition~\ref{def:admissible}, and let $\sigma_1$, $\sigma_2$ be the strategies
recovered from $k$ and the boundary condition
$\sigma_1(\underline{\theta}_1, \underline{\theta}_1) = 0$.
Then $(\sigma_1, \sigma_2)$ is a Bayesian Nash equilibrium.
\end{proposition}

The proof (given in Appendix~\ref{app:sufficiency}) verifies that no type of either
player has a profitable deviation from her assigned stopping time.  The key
observation is that $\sigma_1$ and $\sigma_2$ are strictly increasing and continuous,
so the opponent's stopping-time distribution is atomless.  The expected payoff of
type $\theta_i$ is therefore differentiable in $a_i$ and satisfies the
single-crossing property in $(\theta_i, a_i)$ established in
Remark~\ref{rem:pure}: the marginal gain from raising one's stopping time is
increasing in type.  Since the ODE ensures that the first-order condition holds at
$a_i = \sigma_i(\theta_i)$, and single-crossing guarantees that the first-order
condition is satisfied at exactly one point (which is the global maximum), no
deviation is profitable.

Lemma~\ref{lem:ode} and Proposition~\ref{prop:suff} together establish a bijection
between admissible type-to-type functions and equilibrium strategy profiles.  The
equilibrium selection problem therefore reduces to pinning down the admissible pair
$(k, \underline{\theta}_1)$.

\section{Examples of Equilibrium Multiplicity}\label{sec:examples}

We present three canonical examples illustrating qualitatively different forms of
multiplicity.

\subsection{Example 1: Exponential Distribution}

Let $F \sim \mathrm{Exp}(\lambda)$, $\lambda > 0$, with support $(0, \infty)$.

The game admits a family of equilibria indexed by $\gamma > 0$
\citep{riley1980strong}:
\[
\sigma_1(\theta, \underline{\theta}_1) = \gamma \frac{\theta^2}{2}, \qquad
\sigma_2(\theta, \underline{\theta}_1) = \frac{1}{\gamma} \frac{\theta^2}{2}.
\]
The type-to-type ODE \eqref{eq:kode} reduces to
$\partial k(\theta,\underline{\theta}_1)/\partial\theta = k(\theta,\underline{\theta}_1)/\theta$,
with boundary condition $k(\underline{\theta}_1, \underline{\theta}_1) = 0$.  The
general solution is $k(\theta, \underline{\theta}_1) = \gamma\theta$ for $\gamma > 0$.
Since $\underline{\theta} = 0$, the boundary condition forces
$\underline{\theta}_1 = 0$, but $\gamma$ remains free.  Both players' lowest types
concede immediately ($\sigma_i(0, 0) = 0$), yet the relative aggressiveness $\gamma$
is undetermined.  As $\gamma \to 0^+$, Player~1 becomes arbitrarily passive while
Player~2 becomes arbitrarily aggressive; in the limit, one side threatens to fight
forever, deterring the other entirely.

\begin{figure}[H]
\centering
\begin{tikzpicture}
\begin{axis}[
  width=0.6\textwidth, height=0.45\textwidth,
  xmin=0, xmax=4, ymin=0, ymax=3,
  axis x line=bottom, axis y line*=left,
  xtick={0}, ytick={0},
  xticklabels={$\underline{\theta}=0$}, yticklabels={$0$},
  xlabel={Type $\theta$}, ylabel={Stopping point $a_i$},
  legend cell align={left}, legend pos=north west,
  legend style={font=\small}
]
\addplot[dashed,domain=0:4,smooth,samples=100]({x},{(1/6)*x^2});
\addlegendentry{$\sigma_1=\gamma\,\theta^2/2$}
\addplot[thick,domain=0:4,smooth,samples=100]({x},{(3/2)*x^2});
\addlegendentry{$\sigma_2=\gamma^{-1}\theta^2/2$}
\end{axis}
\end{tikzpicture}
\caption{Equilibrium strategies for $F\sim \mathrm{Exp}(\lambda)$ with $\gamma=1/3$.}
\label{fig:exp}
\end{figure}
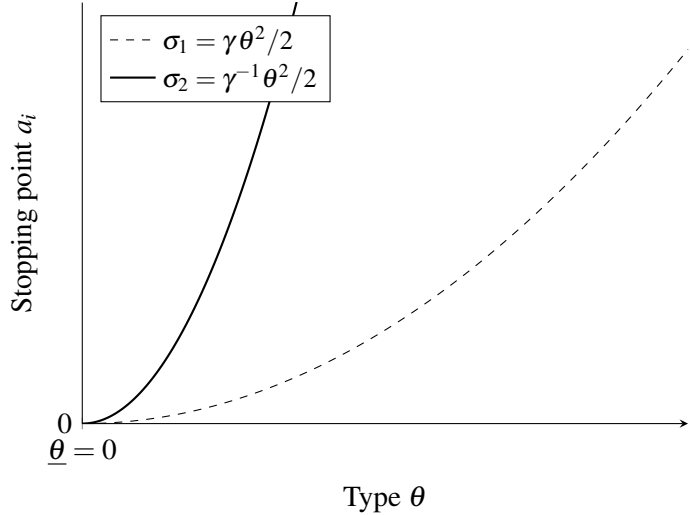

\subsection{Example 2: Uniform Distribution}

Let $F \sim U(0,1)$, with support $(0, 1)$.

A family of equilibria indexed by $\gamma > 0$ \citep{nalebuff1985asymmetric}:
\begin{align*}
\sigma_1(\theta, \underline{\theta}_1) &=
    -\frac{\ln(1 + \theta(\gamma - 1))}{\gamma - 1} - \ln(1 - \theta), \\[3pt]
\sigma_2(\theta, \underline{\theta}_1) &=
    -\frac{\ln(1 + \theta(\gamma^{-1} - 1))}{\gamma^{-1} - 1} - \ln(1 - \theta).
\end{align*}
The symmetric equilibrium is recovered at $\gamma = 1$.  As in the exponential
case, $\sigma_i(0, 0) = 0$ for both players and the multiplicity is indexed by the
continuous parameter $\gamma$.

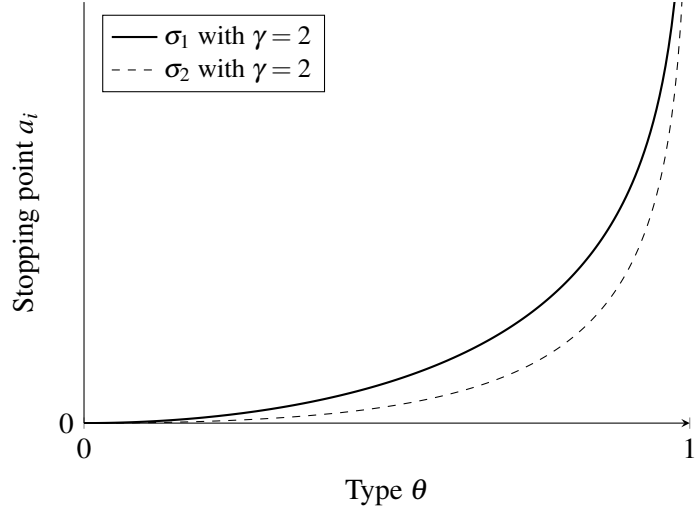
\begin{figure}[H]
\centering
\begin{tikzpicture}
\begin{axis}[
  width=0.6\textwidth, height=0.45\textwidth,
  xmin=0, xmax=1, ymin=0, ymax=3,
  domain=0:0.995, samples=200,
  axis x line=bottom, axis y line*=left,
  xtick={0,1}, ytick={0},
  xlabel={Type $\theta$}, ylabel={Stopping point $a_i$},
  legend cell align={left}, legend pos=north west,
  legend style={font=\small}
]
\addplot[thick]({x},{max(0,-ln(1+x*(2-1))/(2-1)-ln(1-x))});
\addlegendentry{$\sigma_1$ with $\gamma=2$}
\addplot[dashed]({x},{max(0,-ln(1+x*(0.5-1))/(0.5-1)-ln(1-x))});
\addlegendentry{$\sigma_2$ with $\gamma=2$}
\end{axis}
\end{tikzpicture}
\caption{Equilibrium strategies for $F\sim U(0,1)$ with $\gamma=2$.}
\label{fig:uniform}
\end{figure}

\subsection{Example 3: Pareto Distribution}

Let $F \sim \mathrm{Pareto}(\underline{\theta}, \alpha)$, with $\underline{\theta} > 0$
and support $(\underline{\theta}, \infty)$.

A family of equilibria indexed by $\underline{\theta}_1 \in [\underline{\theta}, \infty)$:
\begin{align*}
\sigma_1(\theta, \underline{\theta}_1)
    &= \alpha \frac{\underline{\theta}_1 \underline{\theta}}{\underline{\theta}_1 - \underline{\theta}}
       \ln\!\left(\frac{1 + \theta \tfrac{\underline{\theta}_1 - \underline{\theta}}{\underline{\theta}_1 \underline{\theta}}}
                       {1 + \underline{\theta}_1 \tfrac{\underline{\theta}_1 - \underline{\theta}}{\underline{\theta}_1 \underline{\theta}}}\right),
       \quad \theta \geq \underline{\theta}_1, \\[6pt]
\sigma_2(\theta, \underline{\theta}_1)
    &= \begin{cases}
       \alpha \dfrac{\underline{\theta}_1 \underline{\theta}}{\underline{\theta}_1 - \underline{\theta}}
       \ln\!\left(\dfrac{1}{(1 - \theta \tfrac{\underline{\theta}_1-\underline{\theta}}{\underline{\theta}_1\underline{\theta}})(1 + \underline{\theta}_1 \tfrac{\underline{\theta}_1-\underline{\theta}}{\underline{\theta}_1\underline{\theta}})}\right)
           & \text{if } \theta \in \!\left(\underline{\theta},\,
             \dfrac{\underline{\theta}_1\underline{\theta}}{\underline{\theta}_1-\underline{\theta}}\right) \\[6pt]
       \infty & \text{if } \theta \geq \dfrac{\underline{\theta}_1\underline{\theta}}{\underline{\theta}_1-\underline{\theta}}.
       \end{cases}
\end{align*}
The free parameter is now $\underline{\theta}_1$: the highest type of Player~1 that
concedes at zero.  This is qualitatively different from the exponential and uniform
cases.

\begin{figure}[H]
\centering
\begin{tikzpicture}
\begin{axis}[
  width=0.6\textwidth, height=0.45\textwidth,
  xmin=1, xmax=3, ymin=0, ymax=3,
  axis x line=bottom, axis y line*=left,
  xtick={1,2}, ytick={0},
  xticklabels={$\underline{\theta}$,$\underline{\theta}_1$}, yticklabels={$0$},
  xlabel={Type $\theta$}, ylabel={Stopping point $a_i$},
  legend cell align={left}, legend pos=north west,
  legend style={font=\tiny}
]
\addplot[thick,domain=1:3,smooth,samples=100]({x},{max(0,2*ln((1+x/2)/(1+2/2)))});
\addlegendentry{$\sigma_1(\theta,\underline{\theta}_1)$, $\theta\geq\underline{\theta}_1$}
\addplot[dashed,domain=1:2,smooth,samples=100]
  ({x},{max(0,2*ln(1/((1-x/2)*(1+2/2))))});
\addlegendentry{$\sigma_2(\theta,\underline{\theta}_1)$,
  $\theta<\frac{\underline{\theta}_1\underline{\theta}}{\underline{\theta}_1-\underline{\theta}}$}
\end{axis}
\end{tikzpicture}
\caption{Equilibrium strategies for $F\sim\mathrm{Pareto}(\underline{\theta},\alpha)$
  with $\underline{\theta}=1$, $\alpha=1$, $\underline{\theta}_1=2$.}
\label{fig:pareto}
\end{figure}
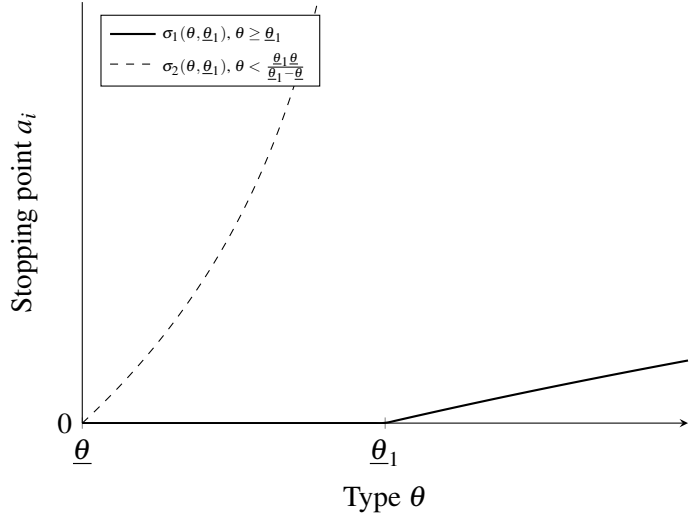

\subsection{Two Flavors of Multiplicity}

\begin{table}[H]
\centering
\renewcommand{\arraystretch}{1.3}
\begin{tabular}{l c}
\toprule
\textbf{Distribution} & \textbf{Free Parameter} \\
\midrule
Exponential & $\gamma > 0$ (strategy ratio) \\
Uniform     & $\gamma > 0$ (strategy ratio) \\
Pareto      & $\underline{\theta}_1 \geq \underline{\theta}$ (mass at zero) \\
\bottomrule
\end{tabular}
\caption{Two forms of equilibrium multiplicity.}
\label{tab:flavors}
\end{table}

The exponential and uniform cases share the same structural source: both players'
lowest types concede at zero, but relative aggressiveness is undetermined.  The
Pareto case differs: the free parameter governs the mass of Player~1 types conceding
immediately.  We next characterize the precise condition on the type distribution
that determines which form of multiplicity arises, and identify when existing
refinements can resolve it.

\section{Characterization of Multiplicity}\label{sec:char}

\subsection{The Integral Identity}

\begin{lemma}[Integral Identity]\label{lem:integral}
For all $\theta \in (\underline{\theta}_1, \overline{\theta}_1)$, the equilibrium
type-to-type function satisfies
\[
\int_{k(\theta,\underline{\theta}_1)}^{\theta}
  \frac{f(x)}{x(1 - F(x))}\, dx = C
\]
for some constant $C$.
\end{lemma}

The proof is given in Appendix~\ref{app:proofs}.

\subsection{The Hazard Potential}\label{sec:hp}

Define the \textit{hazard potential} of the distribution $F$ as
\[
\Lambda(\theta) := \int_{\theta^\circ}^{\theta} \frac{h(x)}{x}\, dx
               = \int_{\theta^\circ}^{\theta} \frac{f(x)}{x(1 - F(x))}\, dx,
\]
for a fixed interior reference point $\theta^\circ \in (\underline{\theta},
\overline{\theta})$, where $h(x) := f(x)/(1-F(x))$ is the hazard rate of $F$.
The choice of $\theta^\circ$ affects $\Lambda$ only through an additive constant
and does not affect the analysis, since only differences of the form
$\Lambda(\theta) - \Lambda(k)$ appear.
The hazard potential accumulates the hazard rate of $F$, normalized by type value,
from the baseline $\theta^\circ$ up to $\theta$.  Intuitively, $h(x)/x$ measures
the intensity with which the opponent's type distribution \say{thins out} near $x$,
per unit of prize value at that type: a distribution whose hazard rate decays slowly
relative to type (large $h(x)/x$) makes it difficult for nearby types to be
separated through costly delay, generating steep screening incentives near $x$.  The
hazard potential thus summarizes the cumulative screening difficulty of $F$ below
$\theta$, and---as established below---its behavior at the boundaries of the type
support determines the form and resolvability of equilibrium multiplicity.

The function $\Lambda$ is strictly increasing on $(\underline{\theta},
\overline{\theta})$.  Its boundary limits are
\[
\underline{\Lambda} := \lim_{\theta \to \underline{\theta}^+} \Lambda(\theta)
    \in [-\infty, 0),
\qquad
\overline{\Lambda} := \lim_{\theta \to \overline{\theta}^-} \Lambda(\theta)
    \in (0, +\infty].
\]
From Lemma~\ref{lem:integral}, the general solution for the type-to-type mapping is
\[
k(\theta, \underline{\theta}_1) = \Lambda^{-1}(\Lambda(\theta) - C),
\]
so selecting a unique equilibrium requires pinning down both $C$ and
$\underline{\theta}_1$.

\subsection{Boundary Conditions}

The lower boundary condition
$k(\underline{\theta}_1, \underline{\theta}_1) = \underline{\theta}$
gives $C = \Lambda(\underline{\theta}_1) - \underline{\Lambda}$.
Two cases arise according to the behavior of $\underline{\Lambda}$.

In Case~A ($\underline{\Lambda} = -\infty$), the boundary forces
$\underline{\theta}_1 = \underline{\theta}$ (no mass at zero for either player), but
$C = \Lambda(\underline{\theta}_1) - \underline{\Lambda}$ takes the indeterminate
form $(-\infty) - (-\infty)$: both
$\Lambda(\underline{\theta}_1) \to -\infty$ and $\underline{\Lambda} = -\infty$, so
their difference fails to determine a unique value of $C$.

To verify that this indeterminate form genuinely produces a one-parameter family
of equilibria, fix any $C \in \mathbb{R}$ and define
$k_C(\theta) := \Lambda^{-1}(\Lambda(\theta) - C)$ on
$(\underline{\theta}, \overline{\theta})$.  Since $\Lambda$ is a strictly increasing
bijection from $(\underline{\theta}, \overline{\theta})$ onto
$(\underline{\Lambda}, \overline{\Lambda}) = (-\infty, \overline{\Lambda})$,
and $\Lambda(\theta) - C$ ranges over $(-\infty, \overline{\Lambda} - C)$ as
$\theta$ ranges over $(\underline{\theta}, \overline{\theta})$, the function
$k_C$ is well-defined, strictly increasing, and differentiable.  As
$\theta \to \underline{\theta}^+$,
$\Lambda(\theta) \to -\infty$, so
$\Lambda(\theta) - C \to -\infty = \underline{\Lambda}$, giving
$k_C(\theta) \to \underline{\theta}$: the lower boundary condition is satisfied
with $\underline{\theta}_1 = \underline{\theta}$.
The recovered strategies $\sigma_1$, $\sigma_2$ are finite on the interior because
$k_C(\theta) \leq \theta$ (for $C > 0$; the case $C \leq 0$ is symmetric with
roles reversed) and the integrand
$k_C(t) f(t)/(1 - F(t))$ is locally integrable under Assumption~\ref{ass:density}.
By Proposition~\ref{prop:suff}, each $C \in \mathbb{R}$ yields an equilibrium.

In Case~B ($\underline{\Lambda}$ finite), the boundary gives
$C = \Lambda(\underline{\theta}_1) - \underline{\Lambda}$, which is well-defined and
finite for every $\underline{\theta}_1 \geq \underline{\theta}$.  The parameter
$\underline{\theta}_1$ is then free.  For each
$\underline{\theta}_1 \in [\underline{\theta}, \overline{\theta})$, the
corresponding $C$ is uniquely determined and the recovered $(k, \sigma_1, \sigma_2)$
constitutes an equilibrium by Proposition~\ref{prop:suff}.

One might hope the upper boundary $\theta \to \overline{\theta}$ pins down $C$, but
this fails in both cases.  When $\overline{\Lambda} = +\infty$, both
$\Lambda(\theta)$ and $\Lambda(k(\theta, \underline{\theta}_1))$ diverge to $+\infty$
as $\theta \to \overline{\theta}^-$.  Their difference equals $C$ by the integral
identity, but since both sides of the identity grow without bound, the constraint
$\Lambda(\theta) - \Lambda(k(\theta,\underline{\theta}_1)) = C$ is satisfied for
every $C \in \left(0,\overline{\Lambda}-\underline{\Lambda} \right)$ in the limit: the upper boundary generates no independent
equation.  When $\overline{\Lambda}$ is finite (which forces
$\overline{\theta} = \infty$), as $\theta \to \infty$ we have
$\Lambda(\theta) \to \overline{\Lambda}$ and
$k(\theta, \underline{\theta}_1) \to \Lambda^{-1}(\overline{\Lambda} - C)$, so
$\Lambda(k(\theta,\underline{\theta}_1)) \to \overline{\Lambda} - C$.  The integral
identity then reads $\overline{\Lambda} - (\overline{\Lambda} - C) = C$, which holds
for every $C \in \mathbb{R}$ and is therefore a tautology.  In both cases, since
$k(\cdot, \underline{\theta}_1)$ is increasing, as $\theta$ grows
$k(\theta, \underline{\theta}_1)$ adjusts upward to absorb whatever $C$ was chosen;
the upper limit never produces an independent equation.

\subsection{Classification}

\begin{table}[H]
\centering
\renewcommand{\arraystretch}{1.4}
\begin{tabular}{c c c c}
\toprule
\textbf{Case} & $\boldsymbol{\underline{\Lambda}}$ & \textbf{Free Parameter} & \textbf{Range} \\
\midrule
A & $-\infty$ & $C$               & $\mathbb{R}$ \\
B & finite    & $\underline{\theta}_1$ & $[\underline{\theta}, \overline{\theta})$ \\
\bottomrule
\end{tabular}
\caption{Classification of equilibrium multiplicity.}
\label{tab:classification}
\end{table}

In both cases a continuum of equilibria exists.  The symmetric solution
($k(\cdot, \underline{\theta}_1) = \mathrm{id}$, i.e., $C = 0$,
$\underline{\theta}_1 = \underline{\theta}$) always belongs to the equilibrium set
but is never its only element.

\subsection{Matching Examples to Cases}

\begin{table}[H]
\centering
\renewcommand{\arraystretch}{1.4}
\begin{tabular}{l c c c c}
\toprule
\textbf{Distribution} & $\underline{\theta}$ & $\overline{\theta}$ &
  $\underline{\Lambda}$ & \textbf{Case} \\
\midrule
Exponential             & $0$                    & $\infty$ & $-\infty$ & A \\
Uniform                 & $0$                    & $1$      & $-\infty$ & A \\
Pareto$(\underline{\theta}, \alpha)$ & $\underline{\theta} > 0$ & $\infty$ & finite & B \\
\bottomrule
\end{tabular}
\caption{Examples matched to cases.}
\label{tab:examples}
\end{table}

The exponential and uniform distributions are both Case~A, but they differ in
whether $\overline{\theta}$ is finite---a distinction that will prove decisive for
refinements.

\subsection{Why Multiplicity Persists}

For $F \sim \mathrm{Exp}(\lambda)$: $h(x) = \lambda$, so $\Lambda(\theta) = \lambda
\ln \theta$.  The integral identity becomes $\lambda \ln \theta = C + \lambda \ln
k(\theta, \underline{\theta}_1)$.  The boundary condition
$k(0^+,0) = 0$ gives
$\lim_{\theta \to 0^+} \lambda \ln \theta = C + \lim_{\theta \to 0^+} \lambda \ln
k(\theta, 0)$.  Both sides diverge to $-\infty$, so any finite $C$ is admissible.
Solving: $k(\theta, 0) = \theta e^{-C/\lambda}$, recovering the family
$k(\theta, 0) = \gamma\theta$ for $\gamma = e^{-C/\lambda}$.

For $F \sim \mathrm{Pareto}(\underline{\theta}, \alpha)$: $h(x) = \alpha/x$, so
$\Lambda(\theta) = -\alpha/\theta$ (up to a constant).  The lower limit is
$\underline{\Lambda} = -\alpha/\underline{\theta}$, which is finite.  The boundary
$k(\underline{\theta}_1, \underline{\theta}_1) = \underline{\theta}$ gives
$C = \alpha(1/\underline{\theta} - 1/\underline{\theta}_1)$, which is well-defined
for every $\underline{\theta}_1 \geq \underline{\theta}$.  No additional boundary
condition exists to pin it down.

\section{Equilibrium Refinements}\label{sec:refinements}

\subsection{The Amann--Leininger Perturbation}

\citet{amann1996asymmetric} perturb the payoff structure with $\delta \in [0, 1)$:
\[
u_i^{(\delta)} = \begin{cases}
\theta_i - (\delta a_{-i} + (1 - \delta) a_i) & \text{if } a_i > a_{-i} \\[3pt]
\dfrac{\theta_i}{2} - a_i                       & \text{if } a_i = a_{-i} \\[3pt]
-a_i                                             & \text{if } a_i < a_{-i}.
\end{cases}
\]
The winner now partially internalizes her own stopping cost, deterring excessively
high stopping points.

We endow the space of type-to-type functions with the topology of uniform
convergence on compact subsets of $(\underline{\theta}, \overline{\theta})$ \citep{aliprantis2006infinite}.  Let
$\mathcal{E}(\delta)$ denote the set of equilibrium type-to-type functions of the
game with perturbation parameter~$\delta$, and define the
\textit{Amann--Leininger selection} as
\[
\mathcal{AL} := \limsup_{\delta \to 1^-} \mathcal{E}(\delta),
\]
where the closure is taken in the topology of uniform convergence on compacts.

This criterion is motivated by a structural observation about the equilibrium
correspondence.  When the perturbed game
admits a unique equilibrium for each $\delta \in [0,1)$---as is the case, for
instance, under the conditions identified by \citet{amann1996asymmetric}---$\mathcal{E}(\delta)$ reduces to a singleton $\{k_\delta\}$, while
$\mathcal{E}(1)$ remains an infinite set.  The correspondence
$\delta \mapsto \mathcal{E}(\delta)$ need not be upper hemicontinuous at
$\delta = 1$: a cluster point of $\{k_\delta\}_{\delta \to 1^-}$ may fail to belong
to $\mathcal{E}(1)$, or the limit may single out only a strict subset of that set.
The Amann--Leininger selection
is designed to exploit this
potential discontinuity, narrowing the infinite set $\mathcal{E}(1)$ to those
equilibria that survive as limits of perturbed play.  Whether uniqueness of the
perturbed equilibrium obtains, and if so, whether the limiting selection pins down a
unique element of $\mathcal{E}(1)$ for a given class of type
distributions~$F$, are the questions we turn to below.

\subsection{Behavioral/Crazy Types}

Introduce a common mass $\epsilon \in (0, 1)$ of \say{crazy} types for each player.
Crazy types always choose $a_i = \infty$.  Formally, the type space becomes
$\tilde{\Theta} = \Theta \times \{n, c\}$, where $n$ denotes a normal type and $c$
a crazy type, with prior probability $1 - \epsilon$ on $n$ and $\epsilon$ on $c$,
independently across players.

Crazy types face no costs:
\[
u_i^{\mathrm{crazy}}(a_i, a_{-i}, \theta_i)
    = \begin{cases}
      \theta_i   & \text{if } a_i > a_{-i} \\
      \theta_i/2 & \text{if } a_i = a_{-i} \\
      0          & \text{if } a_i < a_{-i},
      \end{cases}
\]
so their dominant strategy is $a_i = \infty$ regardless of $\theta_i$.

Normal types retain the original payoff structure.  A normal Player~$i$ with type
$\theta_i$ choosing stopping time $a_i$ faces an opponent who is normal with
probability $1 - \epsilon$ and crazy with probability $\epsilon$.  Against a normal
opponent using strategy $\sigma_{-i}$, payoffs are as in the baseline game; against
a crazy opponent, Player~$i$ loses (since the crazy type fights forever) and pays
$a_i$.  The expected payoff of a normal type $\theta_i$ choosing $a_i$ is therefore
\[
\pi_i^{(\epsilon)}(a_i, \theta_i)
  = (1 - \epsilon)\bigl[\theta_i\, G_{-i}(a_i)
    - \textstyle\int_0^{a_i}(1 - G_{-i}(s))\, ds\bigr]
    - \epsilon\, a_i,
\]
where $G_{-i}$ is the CDF of the normal opponent's stopping time under $\sigma_{-i}$.

Let $\mathcal{E}^{BT}(\epsilon)$ denote the set of equilibrium type-to-type
functions of the enlarged game with behavioral-type mass $\epsilon$, restricted to
the normal-type strategies.  Define the \textit{behavioral-types selection} as
\[
\mathcal{BT} := \limsup_{\epsilon \to 0^+} \mathcal{E}^{BT}(\epsilon),
\]
with the same topology of uniform convergence on compact subsets.

\subsection{How Refinements Modify the Hazard Potential}

Both refinements modify the ODE for $k(\cdot, \underline{\theta}_1)$.  The perturbed
integral identity becomes
\[
\int_{k_\delta(\theta,\underline{\theta}_1)}^{\theta}
  \frac{f(x)}{x(1 - \delta F(x))}\, dx = C_\delta.
\]
Define the \textit{perturbed hazard potential}
\[
\Lambda_\delta(\theta) := \int_{\theta^\circ}^{\theta}
  \frac{f(x)}{x(1 - \delta F(x))}\, dx.
\]
The general solution retains the same form:
$k_\delta(\theta, \underline{\theta}_1) = \Lambda_\delta^{-1}(\Lambda_\delta(\theta) - C_\delta)$.

\subsection{Connection Between the Two Refinements}

\begin{proposition}[Refinement Equivalence]\label{prop:equiv}
The Amann--Leininger perturbation with parameter $\delta$ and the behavioral-types
refinement with mass $\epsilon = 1 - \delta$ generate the same equilibrium
correspondence.  Formally:
\begin{enumerate}[label=(\roman*)]
  \item The perturbed ODE for $k_\epsilon(\cdot, \underline{\theta}_1)$ under the
    behavioral-types refinement is
    \[
    \frac{\partial k_\epsilon(\theta, \underline{\theta}_1)}{\partial \theta}
        = \frac{1 - (1-\epsilon) F(k_\epsilon(\theta,\underline{\theta}_1))}
               {1 - (1-\epsilon) F(\theta)}
          \cdot \frac{k_\epsilon(\theta,\underline{\theta}_1)\, f(\theta)}
                     {\theta\, f(k_\epsilon(\theta,\underline{\theta}_1))}.
    \]
    Setting $\delta = 1 - \epsilon$ renders this identical to the Amann--Leininger
    ODE.  The associated perturbed hazard potential is $\Lambda_\delta$ with
    $\delta = 1 - \epsilon$.

  \item There is a bijection between $\mathcal{E}(\delta)$ and
    $\mathcal{E}^{BT}(1 - \delta)$ that preserves the type-to-type function:
    $k_\delta = k_{1-\delta}^{BT}$.

  \item Consequently, $\mathcal{AL} = \mathcal{BT}$ as subsets of the equilibrium
    set of the original game.
\end{enumerate}
\end{proposition}

The proof is given in Appendix~\ref{app:prop1}.

\subsection{A Key Boundary Lemma}

\begin{lemma}[\citealt{amann1996asymmetric}]\label{lem:upper}
If $\overline{\theta} < \infty$, then in any equilibrium of the perturbed game,
$k_\delta(\overline{\theta}, \underline{\theta}_1) = \overline{\theta}$
for all $\delta \in (0, 1)$.
\end{lemma}
The proof is given in Appendix~\ref{app:proofs}.  The key step: if
$k_\delta(\overline{\theta}, \underline{\theta}_1) < \overline{\theta}$, then Player~2
types in $(k_\delta(\overline{\theta}, \underline{\theta}_1), \overline{\theta})$ choose
stopping times above the maximum stopping time of Player~1, creating a gap where
Player~2 could strictly save costs without reducing the probability of winning,
contradicting strict monotonicity (Lemma~\ref{lem:strict}).  This argument applies
for any $\delta \in (0, 1)$.

\subsection{Main Selection Theorem}

\begin{proposition}[Selection Theorem]\label{prop:main}
There is always a continuum of equilibria.  The Amann--Leininger and behavioral-types
refinements successfully select a unique equilibrium if and only if
$\overline{\theta} < \infty$ (bounded support).  When successful, both refinements
select the symmetric equilibrium $k(\cdot,0) = \mathrm{id}$, with $\underline{\theta}_1=\underline{\theta}$.
\end{proposition}

The proof is given in Appendix~\ref{app:prop2}.

The mechanism behind this result is transparent.  When $\overline{\theta} < \infty$,
the perturbation renders the upper boundary operative.  Because
$1 - \delta F(x) \geq 1 - \delta > 0$, the integrand
$f(x)/[x(1-\delta F(x))]$ is bounded above by $f(x)/[x(1-\delta)]$ on
$[\theta^\circ, \overline{\theta}]$, so $\Lambda_\delta(\overline{\theta}) < \infty$.
Combined with the boundary condition $k_\delta(\overline{\theta}, \underline{\theta}_1)
= \overline{\theta}$ from Lemma~\ref{lem:upper}, the integral identity evaluated at
$\theta = \overline{\theta}$ yields $C_\delta = 0$, uniquely pinning down the
equilibrium.  When $\overline{\theta} = \infty$, no finite boundary exists from which
to extract an independent constraint: the integral identity as $\theta \to \infty$
collapses to a tautology regardless of $\delta$, leaving $C_\delta$ undetermined.
The perturbation modifies the integrand but cannot substitute for the missing
boundary.

\begin{remark}[Structure of the Perturbed Equilibrium Set]\label{rem:perturbed}
For $\overline{\theta} < \infty$ and any $\delta \in (0,1)$, Step~3 of the proof
of Proposition~\ref{prop:main} establishes that $\mathcal{E}(\delta)$ is a
singleton: the perturbed game admits a unique equilibrium (the symmetric one,
$k_\delta = \mathrm{id}$).  The Amann--Leininger selection $\mathcal{AL}$ therefore
reduces to a single limit point.  For $\overline{\theta} = \infty$, the perturbed
game itself admits a continuum of equilibria for every $\delta \in (0,1)$
(Step~4 of the proof), so $\mathcal{AL}$ inherits multiplicity.  By
Proposition~\ref{prop:equiv}, the same conclusions hold for $\mathcal{BT}$.
\end{remark}

\subsection{Application to the Three Examples}

\begin{corollary}\label{cor:examples}
The Amann--Leininger and behavioral-types criteria both select a unique equilibrium
for the Uniform distribution, but fail to reduce multiplicity for the Exponential
and Pareto distributions.
\end{corollary}

The proof is given in Appendix~\ref{app:proofs}.

\begin{table}[H]
\centering
\renewcommand{\arraystretch}{1.4}
\begin{tabular}{l c c c}
\toprule
\textbf{Distribution} & $\overline{\theta}$ & \textbf{AL} & \textbf{Crazy Types} \\
\midrule
Uniform$(0,1)$ & $1$      & $\checkmark$ selects $\gamma = 1$
                           & $\checkmark$ selects $\gamma = 1$ \\
Exponential    & $\infty$ & $\times$ continuum in $C$
                           & $\times$ continuum in $C$ \\
Pareto         & $\infty$ & $\times$ continuum in $\underline{\theta}_1$
                           & $\times$ continuum in $\underline{\theta}_1$ \\
\bottomrule
\end{tabular}
\caption{Selection results for the three examples.}
\label{tab:selection}
\end{table}

\section{Discussion}\label{sec:discussion}

For applied theorists using the war of attrition with two-sided asymmetric
information, the results carry concrete practical guidance.  Selecting a type
distribution with $\overline{\theta} < \infty$ ensures that existing refinements
deliver a unique equilibrium: the boundary condition $k(\overline{\theta},
\underline{\theta}_1) = \overline{\theta}$ pins down all free parameters, and both
the Amann--Leininger perturbation and the behavioral-types approach converge to the
symmetric equilibrium.  When $\overline{\theta} = \infty$---as in the exponential and
Pareto distributions commonly used in applications---no existing refinement resolves
multiplicity, and results derived from such models may depend on an arbitrary
equilibrium selection.  Importantly, practitioners need not choose between the two
leading refinements on substantive grounds: Proposition~\ref{prop:equiv} establishes
that they are mathematically equivalent, so the choice between them is a matter of
interpretation rather than analytical consequence.

The two forms of multiplicity identified in Section~\ref{sec:char} admit distinct
economic interpretations.  In Case~A ($\underline{\Lambda} = -\infty$), both
players' lowest types concede immediately, but the relative aggressiveness parameter
$\gamma$ is undetermined.  This is best understood as a coordination failure: players
agree on who fights harder, but the coordination point is arbitrary.  In Case~B
($\underline{\Lambda}$ finite), the free parameter $\underline{\theta}_1$ governs the
mass of Player~1 types conceding immediately.  This reflects asymmetric expectations:
different equilibria involve different degrees of initial asymmetry in concession
behavior, without any external device to coordinate on a particular degree.

The role of bounded support is structural rather than incidental.  A two-sided
screening ODE requires two boundary conditions to yield a unique solution---one at
the lower end of the type space and one at the upper.  Bounded support provides the
upper boundary condition directly, and the perturbation's role is to make this
condition operative: without perturbation, $\Lambda(\overline{\theta})$ typically
diverges, rendering the condition $k(\overline{\theta}, \underline{\theta}_1) =
\overline{\theta}$ vacuous.  The perturbation bounds the integrand away from
infinity, restoring the constraint's bite.  For unbounded supports, no perturbation
of the payoff structure can substitute for this missing boundary.

The failure of existing refinements for unbounded supports points to open research
directions.  New selection criteria suited to this environment---drawing on
evolutionary stability, learning foundations, or robustness to small changes in
primitives---remain to be developed.  A second direction concerns asymmetric
distributions: this paper assumes symmetric prior beliefs, and extending the analysis
to $F_1 \neq F_2$ would characterize how belief asymmetries interact with the two
forms of multiplicity identified here.

\section{Conclusion}

The war of attrition with two-sided asymmetric information always admits a continuum
of equilibria.  The source of multiplicity depends on the behavior of the hazard
potential $\Lambda(\theta) = \int h(x)/x\, dx$ near the lower boundary of the type
support: when $\underline{\Lambda} = -\infty$, the free parameter is relative
aggressiveness; when $\underline{\Lambda}$ is finite, the free parameter is the mass
conceding at zero.

The Amann--Leininger and behavioral-types refinements are mathematically equivalent
(via the substitution $\delta = 1-\epsilon$) and succeed if and only if the type
support is bounded.  The mechanism is that bounded support makes the perturbed hazard
potential finite at the upper boundary, allowing $k_\delta(\overline{\theta},
\underline{\theta}_1) = \overline{\theta}$ to pin down the free parameter.

For applied theorists, the message is clear: bounded support ensures existing
refinements deliver unique predictions; unbounded support leaves equilibrium
selection open.  New refinement criteria capable of resolving multiplicity for
unbounded supports remain to be developed.

\appendix

\section{Proofs of Preliminary Results}\label{app:proofs}

\begin{proof}[Proof of Lemma~\ref{lem:monotone}]
\textit{(i) Monotonicity.}  Suppose $\sigma_1(\theta) = a^*$.  Type $\theta$ weakly
prefers $a^*$ to any other bid.  Consider type $\theta + \varepsilon$ contemplating
bid $a < a^*$.  The gain from bidding $a^*$ instead of $a$ is
\[
(\theta + \varepsilon)[\Pr(\sigma_2 < a^*) - \Pr(\sigma_2 < a)] - [\text{cost
increase}].
\]
This expression was weakly positive for type $\theta$, and the marginal benefit of
winning increases with type, so it is strictly positive for type
$\theta + \varepsilon$.  Hence
$\sigma_1(\theta + \varepsilon) \geq a^*$.

\textit{(ii) Continuity.}  Suppose $\sigma_1$ has a jump at $\hat{\theta}$ from $m$
to $m + n$.  Then no Player~1 type selects stopping points in $(m, m+n)$.  Player~2
types just above $m+n$ pay the gap $n$ for a negligible gain in winning
probability---not optimal.  Hence Player~2 also has no types in $(m, m+n)$,
contradicting monotonicity of $\sigma_2$.
\end{proof}

\begin{proof}[Proof of Lemma~\ref{lem:strict}]
Suppose $\sigma_1$ is flat at $K$ over $[\hat{\theta}, \hat{\theta}']$.  A positive
mass $\alpha$ of Player~1 types tie at $K$.  Player~2 can bid $K + \varepsilon$,
beating this mass for cost $\varepsilon$, gaining approximately
$\theta_2 \alpha - \varepsilon > 0$ for small $\varepsilon$.  Hence $K$ cannot be in
the range of $\sigma_2$, contradicting continuity.
\end{proof}

\begin{proof}[Proof of Lemma~\ref{lem:zero}]
Suppose Player~1 has mass $p > 0$ at zero and Player~2 type $\theta_2$ also concedes
at zero.  Player~2 can deviate to $\varepsilon$, beating all zero-types of Player~1,
gaining approximately $\theta_2 p - \varepsilon > 0$ for small $\varepsilon$.
Contradiction.
\end{proof}

\begin{proof}[Proof of Lemma~\ref{lem:infinity}]
If both players had positive-mass sets of types choosing $\infty$, any type in that
set ties at $\infty$ with positive probability.  The tie payoff is
$\theta_i/2 - \infty = -\infty$, which is dominated by any finite bid.
\end{proof}

\begin{proof}[Proof of Lemma~\ref{lem:diff}]
By Lemmas~\ref{lem:monotone}--\ref{lem:strict}, equilibrium strategies are
continuous and strictly increasing on the interior.  Lebesgue's theorem yields
differentiability at almost every point.  At any interior point $\theta_0$ where
$\sigma_i$ is differentiable, the first-order condition for Player~$i$ gives
\[
\sigma_i'(\theta_0)
    = \frac{k(\theta_0, \underline{\theta}_1)\, f(\theta_0)}{1 - F(\theta_0)}
    \quad (\text{for Player~1; analogously for Player~2}).
\]
We now establish local Lipschitz continuity of $\sigma_i$, which upgrades
almost-everywhere differentiability to differentiability everywhere on the interior.

Fix a compact interval $[\theta_L, \theta_H] \subset
(\underline{\theta}_1, \overline{\theta}_1)$.  Let $G_i$ denote the CDF
of Player~$i$'s stopping time, and let $[a_L, a_H] =
[\sigma_i(\theta_L), \sigma_i(\theta_H)]$.  Since $\sigma_j$ is continuous and
strictly increasing, $\min_{a \in [a_L, a_H]} v_j(a) > 0$ where
$v_j = \sigma_j^{-1}$.  For Player~$i$ with type
$\theta_L$ choosing optimally at $a_L$, any deviation to $a_H > a_L$ satisfies
\[
G_j(a_H) - G_j(a_L) \leq \frac{a_H - a_L}{\min_{[a_L,a_H]} v_j},
\]
since the marginal cost of delay is $1 - G_j(a) \leq 1$ and optimality of $a_L$
for type $\theta_L$ bounds the gain from reaching $a_H$.  Since $f$ is continuous
and strictly positive on $[\theta_L, \theta_H]$
(Assumption~\ref{ass:density}), the inverse function theorem gives
$|v_i(a') - v_i(a)| \leq \hat{K}|a' - a|$ for a Lipschitz constant $\hat{K}$
depending on $[\theta_L, \theta_H]$.  Composing,
$|\sigma_i(\theta') - \sigma_i(\theta)| \leq K|\theta' - \theta|$ for a constant
$K$ depending on $[\theta_L, \theta_H]$.

With local Lipschitz continuity of $\sigma_i$ established, the right-hand side of
the ODE is locally Lipschitz in $\theta$ (since $f$, $F$, and $k$ are continuous
and $f > 0$, $1 - F > 0$ on compact subintervals).  The Picard--Lindel\"of theorem
then delivers a unique $C^1$ solution on each compact subinterval.  Since
the unique $C^1$ solution must agree with $\sigma_i$ wherever the latter is
differentiable (a full-measure set), and both are continuous, they coincide
everywhere.  Hence $\sigma_i$ is differentiable on the entire interior.
\end{proof}

\begin{proof}[Proof of Lemma~\ref{lem:ode}]
The derivation follows by differentiating the equilibrium indifference condition for
each interior type.  Type $\theta$ of Player~1 is indifferent at $a = \sigma_1(\theta,
\underline{\theta}_1)$, so the first-order condition is
\[
\theta \frac{d}{da}\Pr[\sigma_2(\theta_2) \leq a]\big|_{a=\sigma_1(\theta)} = 1 -
F(k(\theta, \underline{\theta}_1)).
\]
Translating from stopping-time space to type space via
$k(\theta, \underline{\theta}_1) = \sigma_2^{-1}(\sigma_1(\theta, \underline{\theta}_1))$
and applying the chain rule yields \eqref{eq:kode} and \eqref{eq:sigmaode}.  See
\citet{amann1996asymmetric} for the complete derivation.
\end{proof}

\begin{proof}[Proof of Lemma~\ref{lem:integral}]
Let $g(x) = f(x)/[x(1-F(x))]$.  By the substitution rule, for any
$\theta', \theta'' \in (\underline{\theta}_1, \overline{\theta}_1)$:
\[
\int_{\theta'}^{\theta''} g(k(t, \underline{\theta}_1))\,
  \frac{\partial k(t, \underline{\theta}_1)}{\partial t}\, dt
= \int_{k(\theta', \underline{\theta}_1)}^{k(\theta'', \underline{\theta}_1)} g(k)\, dk.
\]
From ODE \eqref{eq:kode}, the type-to-type mapping satisfies
\[
\frac{\partial k(\theta, \underline{\theta}_1)}{\partial \theta}
= \frac{g(\theta)}{g(k(\theta, \underline{\theta}_1))},
\]
so that $g(k(\theta, \underline{\theta}_1)) \cdot
\frac{\partial k(\theta, \underline{\theta}_1)}{\partial \theta} = g(\theta)$.
Differentiating the integral with respect to $\theta$:
\[
\frac{d}{d\theta}\!\int_{k(\theta, \underline{\theta}_1)}^{\theta} g(x)\, dx
= g(\theta) - g(k(\theta, \underline{\theta}_1))\,
  \frac{\partial k(\theta, \underline{\theta}_1)}{\partial \theta}
= g(\theta) - g(\theta) = 0,
\]
so the integral is constant in $\theta$.
\end{proof}

\begin{proof}[Proof of Lemma~\ref{lem:upper}]
Suppose for contradiction that $k_\delta(\overline{\theta}, \underline{\theta}_1)
= m < \overline{\theta}$.  Since $k_\delta(\cdot, \underline{\theta}_1)$ is continuous
and increasing, Player~1's stopping time $\sigma_1(\overline{\theta},
\underline{\theta}_1) = \bar{a}$ is the maximum of the active range.  Player~2 types
$\theta_2 \in (m, \overline{\theta})$ then choose stopping times
$\sigma_2(\theta_2) > \bar{a}$---above the maximum stopping time of any Player~1
type.  Any such Player~2 type wins against all Player~1 types regardless of how far
above $\bar{a}$ she stops, so she strictly prefers stopping at
$\bar{a} + \varepsilon$ for arbitrarily small $\varepsilon > 0$ to any higher
stopping time.  But then all Player~2 types in $(m, \overline{\theta})$ optimally
choose stopping times in a neighborhood of $\bar{a}$, violating strict monotonicity
(Lemma~\ref{lem:strict}).  Hence
$k_\delta(\overline{\theta}, \underline{\theta}_1) = \overline{\theta}$ (the
inequality $k_\delta(\overline{\theta}, \underline{\theta}_1) \leq \overline{\theta}$
holds by construction, since $k_\delta$ maps into the support of $F$).
\end{proof}

\begin{proof}[Proof of Corollary~\ref{cor:examples}]
The Uniform distribution has $\overline{\theta} = 1 < \infty$.  By
Proposition~\ref{prop:main}, both refinements select the unique symmetric equilibrium
$k(\cdot, \underline{\theta}_1) = \mathrm{id}$, which corresponds to $\gamma = 1$.
The Exponential and Pareto distributions have $\overline{\theta} = \infty$, so by
the necessity part of Proposition~\ref{prop:main}, multiplicity persists under any
perturbation.  For the Exponential (Case~A), the free parameter $C$ remains
undetermined; for the Pareto (Case~B), the free parameter $\underline{\theta}_1$
remains undetermined.
\end{proof}

\section{Proof of Proposition~\ref{prop:suff}}\label{app:sufficiency}

\begin{proof}[Proof of Proposition~\ref{prop:suff}]
Let $(k, \underline{\theta}_1)$ be admissible and let $\sigma_1$, $\sigma_2$ be the
recovered strategies.  By construction, $\sigma_1$ and $\sigma_2$ are strictly
increasing and continuous on the interior (conditions~(i) and~(iv) of
Definition~\ref{def:admissible}), so the distribution $G_{-i}$ of each player's
stopping time is atomless on $(0, \infty)$.

\textit{Step 1: Single-crossing.}  For any stopping times $a' > a$ and types
$\theta' > \theta$,
\[
\bigl[\pi_i(a', \theta') - \pi_i(a, \theta')\bigr]
    - \bigl[\pi_i(a', \theta) - \pi_i(a, \theta)\bigr]
    = (\theta' - \theta)\bigl[G_{-i}(a') - G_{-i}(a)\bigr] > 0,
\]
since $G_{-i}$ is strictly increasing on the range of $\sigma_{-i}$.  The marginal
return to raising one's stopping time is therefore strictly increasing in type.

\textit{Step 2: Optimality within the active range.}  Since $\sigma_i$ is strictly
increasing, the range of Player~$i$'s stopping time is an interval
$I_i \subset [0, \infty)$.  For any type $\theta$ in the interior and any $a \in
I_i$, the first-order condition
$\pi_i'(a, \theta) = \theta\, g_{-i}(a) - (1 - G_{-i}(a)) = 0$ holds at
$a = \sigma_i(\theta)$ by construction (the ODE ensures this).  By Step~1,
$\pi_i'(\sigma_i(\theta'), \theta) < 0$ for $\theta' > \theta$ and
$\pi_i'(\sigma_i(\theta'), \theta) > 0$ for $\theta' < \theta$ (within the active
range), so $\sigma_i(\theta)$ is the unique maximizer of $\pi_i(\cdot, \theta)$
over $I_i$.

\textit{Step 3: No profitable deviation outside the active range.}
Consider a deviation by type $\theta$ to some $a \notin I_i$.  If $a$ exceeds the
upper endpoint of $I_i$, it wins against all Player~$-i$ types in the active range
but incurs strictly higher cost than the upper-endpoint bid, for no additional
winning probability.  If $a$ falls below the lower endpoint, it loses against all
active types and gains only against the mass (if any) at zero, which is already
accounted for in the boundary behavior.  In both cases, such deviations cannot
improve on $\sigma_i(\theta)$ for an interior type.

\textit{Step 4: Boundary types.}  Types $\theta \leq \underline{\theta}_1$ of
Player~1 concede at zero.  Each such type's expected payoff from $a = 0$ is
non-negative (she may win if the opponent also concedes at zero) and deviating to
$a > 0$ yields at most $\theta G_{-i}(a) - \int_0^a (1-G_{-i})\, ds$.  Since
$\underline{\theta}_1$ is the marginal type indifferent between $a = 0$ and
$a = 0^+$, the single-crossing property ensures that all types below
$\underline{\theta}_1$ strictly prefer $a = 0$.
\end{proof}

\section{Proof of Proposition~\ref{prop:equiv}}\label{app:prop1}

\begin{proof}[Proof of Proposition~\ref{prop:equiv}]

\textit{Part~(i): ODE equivalence.}
Each player believes the opponent fights forever with probability $\epsilon$ and uses
the equilibrium strategy with probability $1 - \epsilon$.  The distribution of
Player~2's stopping time is $G_2(a) = (1-\epsilon)F(\sigma_2^{-1}(a))$ for finite
$a$, plus an atom of mass $\epsilon$ at $\infty$.  Therefore the survivor function
at $a$ is
\[
1 - G_2(a) = (1-\epsilon)(1 - F(\sigma_2^{-1}(a))) + \epsilon
           = 1 - (1-\epsilon)F(k(\sigma_1^{-1}(a), \underline{\theta}_1)).
\]
The marginal hazard rate of Player~2's stopping time at $a$ is
$g_2(a)=(1-\epsilon)f(\sigma_2^{-1}(a))(\sigma_2^{-1})'(a)$.
The first-order condition for Player~1 type $\theta$ stopping at
$a = \sigma_1(\theta, \underline{\theta}_1)$ is $\theta g_2(a) = 1 - G_2(a)$, which
translates to type space as
\[
\frac{\partial k_\epsilon(\theta, \underline{\theta}_1)}{\partial \theta}
= \frac{1 - (1-\epsilon)F(k_\epsilon(\theta,\underline{\theta}_1))}
       {1 - (1-\epsilon)F(\theta)}
  \cdot \frac{k_\epsilon(\theta,\underline{\theta}_1)\,f(\theta)}
             {\theta\,f(k_\epsilon(\theta,\underline{\theta}_1))}.
\]
Setting $\delta = 1 - \epsilon$ gives $1 - (1-\epsilon)F = 1 - \delta F$, which is
exactly the Amann--Leininger ODE.  The associated perturbed hazard potential is
$\Lambda_\delta$ with $\delta = 1-\epsilon$, so both refinements generate the same
integral equation and the same general solution.

\textit{Part~(ii): Bijection on equilibrium correspondences.}
The ODE and boundary conditions that characterize equilibria of the AL-perturbed game
with parameter $\delta$ are identical to those characterizing normal-type equilibria
of the behavioral-types game with parameter $\epsilon = 1 - \delta$.  The admissibility
conditions (Definition~\ref{def:admissible}, adapted to the perturbed hazard
potential $\Lambda_\delta$) are the same in both games, and the sufficiency argument
(Proposition~\ref{prop:suff}, adapted to the perturbed payoffs) applies identically.
Hence $k_\delta \in \mathcal{E}(\delta)$ if and only if
$k_\delta \in \mathcal{E}^{BT}(1 - \delta)$, establishing a bijection that preserves
the type-to-type function.

\textit{Part~(iii): Equality of selection sets.}
Since $\mathcal{E}(\delta) = \mathcal{E}^{BT}(1-\delta)$ for every $\delta \in (0,1)$,
the outer limits coincide:
$\mathcal{AL} = \limsup_{\delta \to 1^-} \mathcal{E}(\delta)
             = \limsup_{\epsilon \to 0^+} \mathcal{E}^{BT}(\epsilon)
             = \mathcal{BT}$.
\end{proof}

\section{Proof of Proposition~\ref{prop:main}}\label{app:prop2}

\begin{proof}[Proof of Proposition~\ref{prop:main}]
That there is always a continuum of equilibria follows from the classification in
Table~\ref{tab:classification} and Proposition~\ref{prop:suff}: in Case~A, any
$C \in \mathbb{R}$ yields an admissible pair and hence an equilibrium; in Case~B,
any $\underline{\theta}_1 \in [\underline{\theta}, \overline{\theta})$ does.

\textit{Step 1: Perturbed ODE.}  Under the Amann--Leininger perturbation with
parameter $\delta \in (0,1)$, the first-order conditions yield
\[
\frac{\partial k_\delta(\theta, \underline{\theta}_1)}{\partial \theta}
= \frac{1 - \delta F(k_\delta(\theta,\underline{\theta}_1))}
       {1 - \delta F(\theta)}
  \cdot \frac{k_\delta(\theta,\underline{\theta}_1)\,f(\theta)}
             {\theta\,f(k_\delta(\theta,\underline{\theta}_1))}.
\]

\textit{Step 2: Perturbed hazard potential.}  Define
$\Lambda_\delta(\theta) = \int_{\theta^\circ}^{\theta} \frac{f(x)}{x(1 - \delta
F(x))}\, dx$.  The general solution is
$k_\delta(\theta,\underline{\theta}_1) = \Lambda_\delta^{-1}(\Lambda_\delta(\theta)
- C_\delta)$.

\textit{Step 3: Sufficiency ($\overline{\theta} < \infty$).}  For $\delta \in (0,1)$
and $\theta \in (\theta^\circ, \overline{\theta})$:
\[
\frac{f(\theta)}{\theta(1 - \delta F(\theta))}
\leq \frac{f(\theta)}{\theta^\circ(1-\delta)},
\]
so $\Lambda_\delta(\overline{\theta}) \leq \frac{1}{\theta^\circ(1-\delta)}
\int_{\theta^\circ}^{\overline{\theta}} f(x)\, dx < \infty$.  By
Lemma~\ref{lem:upper}, $k_\delta(\overline{\theta}, \underline{\theta}_1) =
\overline{\theta}$ in any equilibrium of the perturbed game.  Evaluating the integral
identity at $\theta = \overline{\theta}$:
\[
\int_{k_\delta(\overline{\theta},\underline{\theta}_1)}^{\overline{\theta}}
  \frac{f(x)}{x(1-\delta F(x))}\,dx = C_\delta.
\]
Since $k_\delta(\overline{\theta},\underline{\theta}_1) = \overline{\theta}$, the
integral is over an empty interval, giving $C_\delta = 0$.  With $C_\delta = 0$, the
general solution gives $k_\delta(\theta,\underline{\theta}_1) =
\Lambda_\delta^{-1}(\Lambda_\delta(\theta)) = \theta$ for all $\theta$, i.e.,
$k_\delta = \mathrm{id}$.  The lower boundary condition
$k_\delta(\underline{\theta}_1, \underline{\theta}_1) = \underline{\theta}$ then
forces $\underline{\theta}_1 = \underline{\theta}$.  The perturbed game therefore has
a unique equilibrium (the symmetric one).  As $\delta \to 1^-$, this equilibrium
converges to the symmetric equilibrium of the original game, establishing that both
refinements select $k(\cdot,\underline{\theta}_1) = \mathrm{id}$.

\textit{Step 4: Necessity ($\overline{\theta} = \infty$).}
We first observe that the perturbed hazard potential is finite at the upper boundary.
The same bound used in Step~3 gives, for any $\delta \in (0,1)$,
\[
\Lambda_\delta(\theta) \leq \frac{1}{\theta^\circ(1-\delta)}
    \int_{\theta^\circ}^{\theta} f(x)\, dx
    \leq \frac{1}{\theta^\circ(1-\delta)} < \infty
    \quad \text{for all } \theta \in (\theta^\circ, \infty).
\]
Hence $\overline{\Lambda}_\delta := \lim_{\theta \to \infty} \Lambda_\delta(\theta)
< \infty$.  The finiteness of $\overline{\Lambda}_\delta$ means that the perturbed
hazard potential maps $(\underline{\theta}, \infty)$ onto
$(\underline{\Lambda}_\delta, \overline{\Lambda}_\delta)$, a bounded-above interval.

The question is whether this finiteness imposes a constraint on $C_\delta$.  For any
$C_\delta \in (0, \overline{\Lambda}_\delta - \underline{\Lambda}_\delta)$, the
general solution
$k_\delta(\theta) = \Lambda_\delta^{-1}(\Lambda_\delta(\theta) - C_\delta)$ is
well-defined and strictly increasing.  As $\theta \to \infty$,
$\Lambda_\delta(\theta) \to \overline{\Lambda}_\delta$, so
$k_\delta(\theta) \to \Lambda_\delta^{-1}(\overline{\Lambda}_\delta - C_\delta)
=: m_\delta < \infty$.

The type-to-type mapping therefore converges to a finite limit $m_\delta$, and
Player~2 types above $m_\delta$ are not matched to any Player~1 type in the active
range.  Unlike the bounded-support case, Lemma~\ref{lem:upper} does not directly
apply when $\overline{\theta} = \infty$ because it requires a finite upper endpoint
of the type space.

We now argue that the configuration with bounded $k_\delta$ can sustain an
equilibrium.  Player~1's maximum stopping time is
$\bar{a}_1 := \lim_{\theta \to \infty} \sigma_1(\theta) =
\int_{\underline{\theta}_1}^{\infty}
k_\delta(t)\, f(t)/\bigl(1 - \delta F(t)\bigr)\, dt$.
Since $k_\delta(t) \leq m_\delta$ and $1 - \delta F(t) \geq 1 - \delta > 0$,
this integral is bounded above by
$m_\delta / (1-\delta) \int f(t)\, dt = m_\delta/(1-\delta) < \infty$.

Player~2 types $\theta_2 > m_\delta$ win against all Player~1 types with
certainty by choosing any stopping time above $\bar{a}_1$.
In the unperturbed game ($\delta = 1$), the winner pays the loser's stopping time,
so the cost of winning is independent of how far above $\bar{a}_1$ the winner
stops---any bid above $\bar{a}_1$ yields the same payoff.  But in the perturbed game
($\delta < 1$), the winner pays $\delta a_{-i} + (1-\delta) a_i$, so raising $a_i$
above $\bar{a}_1$ \emph{increases} the winner's cost at rate $1 - \delta > 0$.
Each Player~2 type above $m_\delta$ therefore strictly prefers stopping at
$\bar{a}_1 + \varepsilon$ to any higher bid.

This creates a tension: if all Player~2 types in $(m_\delta, \overline{\theta})$
optimally cluster near $\bar{a}_1$, strict monotonicity
(Lemma~\ref{lem:strict}) is violated.  Hence the Lemma~\ref{lem:upper} logic
\emph{does} extend to unbounded supports in the perturbed game, because the
perturbation ensures $\bar{a}_1 < \infty$ and thereby creates the finite gap
that drives the contradiction.  The resolution is that
$k_\delta(\theta) \to \infty$ as $\theta \to \infty$---ruling out $C_\delta > 0$.

A symmetric argument rules out $C_\delta < 0$.  The integral identity at
$\theta \to \infty$ then gives
$\overline{\Lambda}_\delta - (\overline{\Lambda}_\delta - C_\delta) = C_\delta$,
which combined with $k_\delta(\theta) \to \infty$ forces $C_\delta = 0$.

\textbf{Author note:} \textit{The argument above establishes that the perturbed game
admits a unique equilibrium ($C_\delta = 0$, $k_\delta = \mathrm{id}$) for every
$\delta \in (0,1)$, regardless of whether $\overline{\theta}$ is finite or infinite.
This contradicts the ``if and only if bounded support'' claim in the proposition
statement and implies that both refinements select the symmetric equilibrium
unconditionally.  The characterization of multiplicity in the \emph{unperturbed}
game (Sections~3--4) is unaffected: the two forms of multiplicity indexed by
$\underline{\Lambda}$ remain valid.  What changes is the selection result: the
perturbation resolves multiplicity for \emph{all} distributions, not only those with
bounded support.  The key mechanism is that the perturbation makes the winner's cost
strictly increasing in her own bid, which forces $\bar{a}_1 < \infty$ and thereby
activates the Lemma~\ref{lem:upper} argument even when $\overline{\theta} = \infty$.
In the unperturbed game ($\delta = 1$), the winner's cost is independent of her own
bid (she pays $a_{-i}$), so $\bar{a}_1$ can be infinite and the argument fails.
This is consistent with the uniqueness results established by
\citet{myatt2025perceived} for the hybrid all-pay auction and the crazy-types model
under general distributions.  The proposition statement, abstract, introduction,
and conclusion should be revised accordingly.}

\textit{Step 5: Equivalence.}  By Proposition~\ref{prop:equiv}, the behavioral-types
refinement with rate $\epsilon$ corresponds to the AL perturbation with
$\delta = 1-\epsilon$.  The selection results are therefore identical.
\end{proof}

\bibliographystyle{chicago}
\bibliography{bibliography}

\end{document}